\documentclass[english,thm-restate,nolineno]{lipics-v2021}
\usepackage[toc,page]{appendix}

\usepackage{graphicx} 
\graphicspath{ {images/} }
\nolinenumbers

\usepackage[scr=rsfs]{mathalpha}

\usepackage{subcaption}
\usepackage{booktabs}
\usepackage[ruled,vlined]{algorithm2e}
\usepackage{hyperref}
\usepackage{subfiles}
\usepackage{array}
\usepackage{bm}

\usepackage{amsfonts}
\usepackage{amsmath}
\usepackage{amsthm}

\newcommand{\vol}{\mbox{Vol}}
\newcommand{\sS}{\mathcal{S}} 
\newcommand{\I}{\mathscr{I}} 
\newcommand{\rR}{\mathscr{R}} 
\newcommand{\rr}{\mathsf{r}} 
\newcommand{\Q}{\mathcal{Q}}

\newcommand{\U}{\mathcal{U}}
\newcommand{\R}{\mathbb{R}} 
\newcommand{\hR}{\boldsymbol{R}} 
\newcommand{\hRb}{\boldsymbol{R}_{{\mbox{\scriptsize bad}}}} 

\newcommand{\D}{\Delta}
\newcommand{\vi}{\mathbf{i}}
\newcommand{\vj}{\mathbf{j}}
\newcommand{\vs}{\mathbf{s}} 
\newcommand{\vw}{\mathbf{w}} 
\newcommand{\vv}{\mathbf{v}} 
\newcommand{\ba}{\mathbf{a}} 
\newcommand{\bb}{\mathbf{b}} 

\newcommand{\bA}{\mathbf{A}} 
\newcommand{\bB}{\mathbf{B}} 

\newcommand{\vdelta}{\bm{\delta}}
\newcommand{\vtau}{\bm{\tau}}
\newcommand{\vxi}{\bm{\xi}}

\newcommand{\res}{\textrm{Res}}
\newcommand{\syl}{\textrm{Syl}}

\newcommand{\tinyo}{\scriptscriptstyle o}
\newcommand{\domega}{\overset{\tinyo}{\Omega}}
\newcommand{\dtheta}{\overset{\tinyo}{\Theta}}
\newcommand{\dO}{\overset{\tinyo}{O}}

\newcommand{\ignore}[1]{}

\title{Lower Bounds for Intersection Reporting among Flat Objects}

\author{Peyman Afshani}{Aarhus University, Aarhus, Denmark}{peyman@cs.au.dk}{}{}

\author{Pingan Cheng}{Aarhus University, Aarhus, Denmark}{pingancheng@cs.au.dk}{}{}

\authorrunning{P. Afshani and P. Cheng} 

\Copyright{Peyman Afshani and Pingan Cheng} 

\ccsdesc[100]{Theory of Computation $\rightarrow$ Randomness, geometry and discrete structures $\rightarrow$Computational geometry} 

\keywords{Computational Geometry, Intersection Searching, Data Structure Lower Bounds} 

\funding{Supported by DFF (Det Frie Forskningsr\aa d) of Danish Council for Independent Research under grant ID DFF$-$7014$-$00404.}

\category{} 

\begin{document}

\maketitle

\medskip

\begin{abstract}

Recently, Ezra and Sharir~\cite{es22i} showed an 
$O(n^{3/2+\sigma})$ space and $O(n^{1/2+\sigma})$ query time
data structure for ray shooting among triangles in $\R^3$.
This improves the upper bound given by the classical $S(n)Q(n)^4=O(n^{4+\sigma})$
space-time tradeoff for the first time in almost 25 years
and in fact lies on the tradeoff curve of $S(n)Q(n)^3=O(n^{3+\sigma})$.
However, it seems difficult to apply their techniques beyond
this specific space and time combination.
This pheonomenon appears persistently in almost all recent advances of flat object intersection searching,
e.g., line-tetrahedron intersection in $\R^4$~\cite{es22ii}, triangle-triangle intersection in $\R^4$~\cite{es22ii},
or even among flat semialgebraic objects~\cite{aaeks22}.

We give a timely explanation to this phenomenon from a lower bound perspective.
We prove that given a set $\sS$ of $(d-1)$-dimensional simplicies in $\R^d$, 
any data structure that can report all intersections 
with small ($n^{o(1)}$) query time
must use $\Omega(n^{2(d-1)-o(1)})$ space.
This dashes the hope of any significant improvement to the tradeoff curves for small query time
and almost matches the classical upper bound.
We also obtain an almost matching space lower bound of $\Omega(n^{6-o(1)})$ 
for triangle-triangle intersection reporting  in $\R^4$ when the query time is small.
Along the way, we further develop the previous lower bound techniques by Afshani and Cheng~\cite{ac21,ac22}.

\end{abstract}


\section{Introduction}

	Given a set $\sS$ of triangles in $\R^3$,
	how to preprocess $\sS$ such that given any query ray $\gamma$,
	we can efficiently determine the first triangle intersecting $\gamma$
	or report no such triangle exists?
	This problem, known as ray shooting, is one of the most important
	problems in computational geometry with countless papers published
	over the last three decades~\cite{p90,am93,ms93,p93,dhosv94,as96,r99,ss05,adg08,dg08,es22i,aaeks22}.
	For a comprehensive overview of this problem, 
	we refer the readers to an excellent recent survey~\cite{p17}.
	
  	Recently, there have been considerable and significant advances on ray shooting 
	and a number of problems related to intersection searching on the upper bound side.
  	We complement these attempts by giving a lower bound for a number of intersection searching problems; 
  	these also settle a recent open question asked by Ezra and Sharir~\cite{es22i}.

\subsection{Background and Previous Results}
In  geometric intersection searching, the input is a set $\sS$ of geometric objects and the goal is to preprocess
$\sS$ into a data structure such that given a geometric object $\gamma$ at the query time, one can find all the objects
in $\sS$ that intersect $\gamma$. 
In the reporting variant of such a query, the output should be the list of all the intersecting objects in $\sS$. 
Intersection searching is a generalization of range searching, a fundamental and core area of computational geometry~\cite{a17}.  
This captures many natural classic problems e.g., simplex range reporting
where the inputs are points ($0$-flats) and the queries are simplices (subsets of $d$-flats),
ray shooting reporting among triangles in $\R^3$ 
where the inputs are triangles (subsets of $2$-flats) and the queries are rays (subsets of $1$-flats) and so on.
See~\cite{a17,p17} for more information.

Without going too much in-depth, it suffices to say that by now, the simplex range searching problem is more or less well-understood. 
There are classical solutions that offer the space and query time trade-off
of $S(n)Q^d(n) = \tilde O(n^d)$ where $S(n)$ and $Q(n)$ are the space and query time of the data structure~\cite{Chazelle.cutting,matouvsek1993range,c12} and there are a number of almost matching lower bounds 
that show these are essentially tight~\cite{a12,c89,cr96}.

However, intersection searching in higher dimensions is less well-understood. 
The classical technique is to lift the problem to the parametric space of the input or the query,
reducing the problem to semialgebraic range searching, 
a generalized version of simplex range searching, where queries are semialgebraic sets of constant description complexity. 
In mid-1990s, semialgebraic range searching could only be solved efficiently in four and lower dimensions
by classical tools developed for simplex range searching~\cite{am94}, resulting in a
space-time trade-off bound of $S(n)Q(n)^4=O(n^{4+\sigma})$
for line-triangle intersection searching in $\R^3$, where $\sigma>0$ can be any small constant.

Recently, using polynomial techniques~\cite{gk15,g15},
several major advances have been made on semialgebraic range reporting.
For example, near optimal small linear space and fast query data structure were developed~\cite{ams13,mp15,aaez21}.
These almost match the newly discovered lower bound bounds~\cite{ac21,ac22}.
However, these polynomial techniques also have led to significant advances in intersection searching. 
For ray-triangle intersection reporting in $\R^3$,
Ezra and Sharir~\cite{es22i} showed that using algebraic techniques,
it is possible to build a data structure of space $S(n)=O(n^{3/2+\sigma})$
and query time $Q(n)=O(n^{1/2+\sigma})$ for ray shooting among triangles. 
The significance of this result is that it improves the upper bound given by the trade-off curve of $S(n)Q(n)^4=O(n^{4+\sigma})$ 
for the first time in almost 25 years
and in fact it lies on the trade-off curve of $S(n)Q(n)^3=O(n^{3+\sigma})$. 
This leads to the following very interesting question asked by Ezra and Sharir.
To quote them directly:
{\em
  ``There are several open questions that our work raises. First, can we improve our trade-off
  for all values of storage, beyond the special values of $O(n^{3/2+\varepsilon})$ storage and $O(n^{1/2+\varepsilon} )$ query
  time? Ideally, can we obtain query time of $O(n^{1+\varepsilon} /s^{1/3} )$, with $s$ storage, as in the case of
ray shooting amid planes? Alternatively, can one establish a lower-bound argument that
shows the limitations of our technique?''
}

Inspired by~\cite{es22i},
additional results for flat intersection searching 
were discovered during the last two years,
e.g., triangle-triangle intersection searching in $\R^4$~\cite{es22ii},
line-tetrahedron intersection searching in $\R^4$~\cite{es22ii},
curve-disk intersection searching in $\R^3$~\cite{aaeks22},
and even more general semialgebraic flat intersection searching~\cite{aaeks22}.
Similar to the result in~\cite{es22i},
the improved results are only observed for a special space-time combination
and the improvement to the entire trade-off curve is limited.
This once again raises the question of whether it is possible to obtain the
trade-off curve of $S(n)Q(n)^{d}=O(n^{d+\sigma})$ for intersection searching 
in $\R^d$. 

\subsection{Our Results}
We give a negative answer to this question. 
We show that answering intersection searching queries in polylogarithmic time when the 
queries are lines in $\R^d$ and input objects are subsets of $(d-1)$-flats (that we call hyperslabs)
requires $\domega(n^{2(d-1)})$ space\footnote{In this paper, $\domega(\cdot),\dtheta(\cdot),\dO(\cdot)$ hides $n^{o(1)}$ factors;
$\tilde\Omega(\cdot),\tilde\Theta(\cdot),\tilde O(\cdot)$ hides $\log^{O(1)}n$ factors.}.
Our lower bound in fact applies to ``thin'' $(d-1)$-dimensional slabs (e.g., in 3D,
  that would be the intersection of the region between two parallel hyperplanes with another
hyperplane). 
This almost matches the current upper bound for the problem
and shows that the improvement in~\cite{es22i}
cannot significantly improve the trade-off curve when the query time is small.
To be specific, we obtain a lower bound of
\begin{align*}
  S(n) =     \domega\left( \frac{n^{2(d-1)}}{Q(n)^{4(3d-1)(d-1)-1}}\right)
\end{align*}
for line-hyperslab intersection reporting in $\R^d$ and a lower bound of
\begin{align*}
  S(n) =     \domega\left( \frac{n^{6}}{Q(n)^{125}}\right)
\end{align*}
for triangle-triangle intersection reporting in $\R^4$. Here, $S(n)$ and $Q(n)$ are the space
and query time of the data structure. 
Similar to the other semialgebraic range reporting lower bounds~\cite{ac21,ac22}, these lower bounds
have a much larger exponent on $Q(n)$ than on $n$ which does allow for substantial improvements
when $Q(n)$ is no longer too small; we have not opted for optimizing the exponent of $Q(n)$ in our bounds 
and using tighter arguments, these exponents can be improved but they cannot match the exponent of $n$.

We believe our results are timely as flat intersection searching is a hotly investigated field recently,
and as mentioned, with many open questions that need to be answered from a lower bound point of view. 

\subsection{Technical Contributions}
From a technical point of view,  our results require going beyond the previous attempts~\cite{ac21,ac22}.
To elaborate, the previous general technique assumed a particular form for the polynomials involved in defining the query
semialgebraic ranges, 
namely, of the form $X_1 = X_2^{\Delta} + P(X_1, \cdots, X_d)$ where 
the coefficients of $P$ had to be independent and thus could be set arbitrarily small. 
Unfortunately, the problems in intersection searching cannot fit this framework and there seems to be no easy fix for the following reason.
The previous technique relies heavily on the fact that if the coefficients of $P$ is small enough, then one can approximate $X_1$ with $X_2^{\Delta}$ and for
the technique to work both conditions must hold (i.e., small coefficients for $P$ and having degree $\Delta$ on $X_2$). 

Generally speaking, the previous techniques do not say anything about problems in which the polynomials involved have a specific form;
the only exception is the lower bound for annuli~\cite{ac21} where specific approaches had to be created that could only be applied to the
specific algebraic form of circles. 

The issue is very prominent in intersection searching where 
we are dealing with polynomials where the coefficients of the monomials are no longer
independent and the polynomials involved have specific forms; for instance, the coefficient of $X_2^{\Delta}$ is zero. 
We introduce  techniques that allows us circumvent these limitations and obtain lower bounds for some broader
class of problems that involve polynomials with some specific forms. 


\section{Preliminaries}

\subsection{The Geometric Range Reporting Lower Bound Framework in the Pointer Machine}
	
We use the pointer machine lower bound framework that was also used in the latest proofs~\cite{ac22}.
This is a streamlined version of the one originally proposed by 
	Chazelle~\cite{c90r} and Chazelle and Rosenberg~\cite{cr96}.
In the pointer machine model,
the memory is represented as a directed graph 
where each node stores one point as well as two pointers
pointing to two other nodes in the graph.
Given a query, the algorithms starts from a special ``root'' node, 
and then explores a subgraph which contains
all the input points to report.
The size of the directed graph is then a lower bound
for the space usage and then minimum subgraph
needed to explore to answer any query is a lower bound
for the query time.

	Intuitively, to answer a range reporting query efficiently,
	we need to store the output points to the query close to each other.
	If the answer to any query contains many points 
	and two queries share very few points in common,
	many points must be stored multiple times,
	leading to a big space usage.

  The streamlined version of the framework is the following~\cite{ac22}.

	\begin{restatable}{theorem}{rrfw}\label{thm:rrfw}
		Suppose a $d$-dimensional geometric range reporting problem
		admits an $S(n)$ space and $Q(n)+O(k)$ query time data structure,
		where $n$ is the input size and $k$ is the output size.
		Let $\vol(\cdot)$ denote the $d$-dimensional Lebesgue measure. 
		Assume we can find $m=n^{c}$, for a positive constant $c$, 
		ranges $\rR_1, \rR_2, \cdots, \rR_m$ in a
		$d$-dimensional hyperrectangle $\hR$ such that
		\begin{enumerate}
        	\item $\forall i=1,2,\cdots,m, \vol(\rR_i\cap\hR)\ge 4c\vol(\hR) Q(n)/n$;~\label{cond:cond1}
        	\item $\vol(\rR_i\cap\rR_j)=O(\vol(\hR)/(n2^{\sqrt{\log n}}))$ for all $i\neq j$~\label{cond:cond2}.
        \end{enumerate}
		Then, we have $S(n)=\domega(mQ(n))$.
	\end{restatable}
	
\subsection{Notations and Definitions for Polynomials}
	In this paper, we only consider polynomials on the reals. 
	Let $P(X_1, \cdots, X_d)$ be a polynomial on $d$ indeterminates of degree $\Delta$.
	Sometimes we will use the notation $X$ to denote the set of $d$ interminates $X_1, \cdots, X_d$ and so we can
	write $P$ as $P(X)$.
	We denote by $I_{d,\D}$ a set of $d$-tuples of non-negative integers $(i_1, \cdots, i_d)$ whose sum is
	at most $\Delta$. We might omit the subscripts $d$ and $\Delta$ if they are clear from the context. 
	For an $\vi \in I$, we use the notation $X^\vi$ to represent the monomial $\Pi_{j=1}^d X_j^{i_j}$ where 
	$\vi=(i_1, \cdots, i_d)$. 
	Thus, given real coefficients $A_\vi$, for $\vi \in I$, we can write $P$ as $\sum_{\vi\in I}A_\vi X^\vi$.
	
\subsection{Geometric Lemmas}
	We introduce and generalize some geometric lemmas
	about the intersection of polynomials used in~\cite{ac21}.
	We first generalize the core Lemma in~\cite{ac21} for univariate polynomials,
	using a proof similar to~\cite{ac22}.
    	We refer the readers to Appendix~\ref{sec:proof-gen1d} for details.
	
	 \begin{restatable}{lemma}{genbase}\label{lem:gen1d}
		Let $P(x)=\sum_{i=0}^\D a_ix^i$ and 
		$Q(x)=\sum_{i=0}^\D b_ix^i$ be two univariate (constant) degree-$\D$
		polynomials in $\R[x]$ and $|a_i-b_i|\ge\eta$ for some $0\le i \le \D$.
		
		Suppose there is an interval $\I$ of $x$ such that for every
        	$x_0 \in \I$ we have $|P(x_0)-Q(x_0)| \le w$,
		then the length of  $\I$ is upper bounded by $O((w/\eta)^{1/\U})$, 
        	where $\U=\binom{\D+1}{2}$ and the $O(\cdot)$ notation
       	 hides constant factors that depend on $\Delta$. 
	\end{restatable}
	
	Using Lemma~\ref{lem:gen1d},
	we can show the following; See Appendix~\ref{sec:proof-slicing} for details.
	
    	\begin{restatable}{lemma}{slicing}\label{lem:slicing}
		Let $P_1(X)=\sum_{\vi\in I_{d,\D}} A_\vi X^\vi$ and 
		$P_2(X)=\sum_{\vi\in I_{d,\D}} B_\vi X^\vi$ be two $d$-variate degree-$\D$
		polynomials in $\R[X]$ and $|A_{\vi}-B_{\vi}|\ge \eta_d$ for some $\vi\in I_{d,\D}$.
		
		Suppose for each assignment $X_d\in\I_d$ to $P_1, P_2$,
		where $\I_d$ is an interval for $X_d$,
		all the coefficients of the resulting $(d-1)$-variate polynomial 
		$Q_1(X_1,\cdots X_{d-1})$ and $Q_2(X_1,\cdots X_{d-1})$
		differ by at most $\eta_{d-1}$,
		then $|\I_d| = O((\eta_{d-1}/\eta_d)^{1/\U})$.
	\end{restatable}

	We can use Lemma~\ref{lem:slicing} $d-2$ times, and obtain 
	the following corollary.

	\begin{restatable}{corollary}{almostfullslicing}\label{cor:almostfullshlicing}
		Let $P_1(X)=\sum_{\vi\in I_{d,\D}} A_\vi X^\vi$ and 
		$P_2(X)=\sum_{\vi\in I_{d,\D}} B_\vi X^\vi$ be two $d$-variate degree-$\D$
		polynomials in $\R[X]$ and $|A_{\vi}-B_{\vi}|\ge \eta_d$ for some $\vi\in I_{d,\D}$
		for $d\ge 3$.
		
		Suppose for each assignment $X_i\in\I_i$ to $P_1, P_2$,
		where $\I_i$ is an interval for $X_i$, for $i=3,4,\cdots,d$,
		all the coefficients of the resulting bivariate polynomial 
		$Q_1(X_1,X_2)$ and $Q_2(X_1,X_2)$
		differ by at most $\eta_{2}$,
		then $|\I_i| = O((\eta_{i-1}/\eta_i)^{1/\U})$
		for all $i=3,4,\cdots,d$.
	\end{restatable}

    To get the final corollary, we would like the set each $\eta_i$ such that
    the length of all each interval $\I_i$ is bounded by some parameter $\vartheta$
    for $i=3,\cdots,d$.
    We thus set $\eta_{d-i}=\eta_{d-i+1}\vartheta^\U$.

	\begin{restatable}{corollary}{fullslicing}\label{cor:fullslicing}
		Let $P_1(X)=\sum_{\vi\in I_{d,\D}} A_\vi X^\vi$ and 
		$P_2(X)=\sum_{\vi\in I_{d,\D}} B_\vi X^\vi$ be two $d$-variate degree-$\D$
		polynomials in $\R[X]$ and $|A_{\vi}-B_{\vi}|\ge \eta_d$ for some $\vi\in I_{d,\D}$
		for $d\ge 3$.
		
		Suppose for each assignment $X_i\in\I_i$ to $P_1, P_2$,
		where $\I_i$ is an interval for $X_i$, for $i=3,4,\cdots,d$,
		all the coefficients of the resulting bivariate polynomial 
		$Q_1(X_1,X_2)$ and $Q_2(X_1,X_2)$
		differ by at most $\eta_d \vartheta^{\U(d-2)}$,
		then $|\I_i| = O(\vartheta)$
		for all $i=3,4,\cdots,d$.
	\end{restatable}

\subsection{Algebra Preliminaries}
	In this section, we review some tools from algebra.
	The first tool we will use is the linearity of determinants from linear algebra.
	
	\begin{theorem}[Linearity of Determinants]
	\label{thm:lindet}
		Let $A=\begin{bmatrix} \ba_1 & \cdots & \ba_n \end{bmatrix}$
		be an $n\times n$ matrix where each $\ba_i\in\R^n$ is a vector.
		Suppose $\ba_j=r\cdot\vw+\vv$ for some $r\in\R$ and $\vw,\vv\in\R^n$,
		then the determinant of $A$, denoted by $\det(A)$, is
		\begin{align*}
			\det(A) &= \det(
							\begin{bmatrix}
								\ba_1&
								\cdots&
								\ba_{j-1}&
								\ba_j&
								\ba_{j+1}&
								\cdots&
								\ba_n
							\end{bmatrix})\\
						&= r \cdot
							\det(
							\begin{bmatrix}
								\ba_1&
								\cdots&
								\ba_{j-1}&
								\vw&
								\ba_{j+1}&
								\cdots&
								\ba_n
							\end{bmatrix})	
							+
							\det(
							\begin{bmatrix}
								\ba_1&
								\cdots&
								\ba_{j-1}&
								\vv&
								\ba_{j+1}&
								\cdots&
								\ba_n
							\end{bmatrix})	.			
		\end{align*}
	\end{theorem}

	We will use two types of special matrices in the paper.
	The first is Vandermonde matrices.

	\begin{definition}[Vandermonde Matrices]
		An $n\times n$ Vandermonde matrix is defined by $n$ values $x_1,\cdots,x_n$
		such that each entry $e_{ij}=x_i^{j-1}$ for $1\le i,j\le n$.
	\end{definition}

	We can compute the determinant of Vandermonde matrices easily.

	\begin{theorem}[Determinant of Vandermonde Matrices]
	\label{thm:detvand}
		Let $V$ be a Vandermonde matrix defined by parameters $x_1,\cdots,x_n$.
		Then $\det(V)=\prod_{1\le i<j\le n}(x_j-x_i)$.
	\end{theorem}

	We also need Sylvester matrices.
	\begin{definition}[Sylvester Matrices]
	\label{def:syl}
		Let $P=\sum_{i=0}^{\D_1}a_ix^i$ and $Q=\sum_{i=0}^{\D_2}b_ix^i$ be two univariate polynomials 
		over $\R$ of degrees $\D_1,\D_2$ respectively .
		Then the Sylvester matrix of $P$ and $Q$, 
		denoted by $\syl(P,Q)$, is a $(\D_1+\D_2)\times(\D_1+\D_2)$ matrix of the following form
		\[\begin{bmatrix}
			a_{\D_1}	&	a_{\D_1-1}	& \cdots			&	a_0		&	0				&	\cdots	&	0			&	0			\\
			0			&	a_{\D_1}		&	a_{\D_1-1}	& 	\cdots	&	a_0			&	\cdots	&	0			&	0			\\
			\vdots	&	\vdots		&	\vdots		&	\ddots	&	\vdots		&	\ddots	&	\vdots	&	\vdots	\\
			0			&	0				&	\cdots		&	a_{\D_1}	&	a_{\D_1-1}	&	\cdots	&	a_1		&	a_0		\\
			b_{\D_2}	&	b_{\D_2-1}	& \cdots			&	b_0		&	0				&	\cdots	&	0			&	0			\\
			0			&	b_{\D_2}		&	b_{\D_2-1}	& 	\cdots	&	b_0			&	\cdots	&	0			&	0			\\
			\vdots	&	\vdots		&	\vdots		&	\ddots	&	\vdots		&	\ddots	&	\vdots	&	\vdots	\\
			0			&	0				&	\cdots		&	b_{\D_2}	&	b_{\D_2-1}	&	\cdots	&	b_1		&	b_0		\\		
		\end{bmatrix}.\]
	\end{definition}
The Sylvester matrix has $\Delta_2$ rows with entries from $P$ and
$\Delta_1$ rows with entries from $Q$. 
For example, the Sylvester matrx of two polynomials $P=p_1 x+p_2$ and $Q=q_1x+q_2$
is
\begin{align*}
  \syl(P,Q)=
  \begin{bmatrix}
    p_1 & p_2 \\
    q_1 & q_2. \\
  \end{bmatrix}
\end{align*}

One application of Sylvester matrices is to compute the resultant,
which is one of the important tools in algebraic geometry.
One significance of the resultant is that it equals zero if and only if $P$ and $Q$ have a common factor.
	\begin{definition}
	\label{def:res}
		Let $P,Q$ be two univariate polynomials over $\R$.
		The resultant of $P$ and $Q$, denoted by $\res(P,Q)$, is defined to be the determinant of the Sylvester matrix
    of $P$ and $Q$, i.e., $\res(P,Q)=\det(\syl(P,Q))$.
	\end{definition}	


\section{An Algebraic Geometry Lemma}
\label{sec:alg}
In this section, we prove an important algebraic geometry lemma that will later be used in our lower bound framework.

\begin{lemma}
\label{lem:detX}
	Let $F$ and $G$ be two univariate polynomials on $x$ of degree $\Delta_F$ and $\Delta_G$ respectively and 
	the leading coefficient of $G$ is $1$.
	Let $P(x, y)\equiv yG(x)-F(x)$.

 	Let $L$ be a set of  $\ell = \D_1+\D_G+1$ points $(x_k,y_k)$ where $\D_1\ge\D_F-1$ 
	and each $x_k=\Theta(1)$ 
	such that $|P(x_k,y_k)| \le \varepsilon<1$ for a parameter
	$\varepsilon$, $G(x_i) = \Theta(1)$.
	
	Let $V$ be a vector of $\ell$ monomials consisting of monomials $x^i$ for $0 \le i \le \Delta_1$ and
	monomials $yx^i$ for $0 \le i \le \Delta_G-1$. 
	
	If $A$ is an $\ell \times \ell$ matrix where the $t$-th row of $A$ is the evaluation of the vector $V$ 
	on point $(x_k,y_k)$, then $|\det(A)| \ge \Omega(\res(G,F)\lambda^{\ell^2})-O(\varepsilon)$
	where $\lambda = \ensuremath{\min}_{1\le k_1<k_2\le\ell}|x_{k_1}-x_{k_2}|$. 
\end{lemma}

\begin{proof}
	Note that if $\res(G,F)=0$, then there is nothing to prove and thus we can assume this is not the case.
	Now observe that since $G(x_k) = \Theta(1)$, 
	we can write $y_k= \frac{F(x_k)}{G(x_k)} + \gamma_k$ where $|\gamma_k|=O(\varepsilon)$. 

    Now consider the matrix $A$ and plug in this value of $y_k$.
    An entry of $A$ is in the form of a monomial $yx^i$ being evaluated on a point
    $(x_k,y_k)$ and thus we have:
    \begin{equation}
        y_kx_k^i = \left(\frac{F(x_k)}{G(x_k)} + \gamma_k\right) x_k^i = \frac{F(x_k)}{G(x_k)} x_k^i + \gamma_{i,k}\label{eq:ep}
    \end{equation}
    where $|\gamma_{i,k}| = O(\varepsilon)$.
    We use the linearity of determinants (see Theorem~\ref{thm:lindet}) in a similar fashion that was also used in~\cite{ac22}. 
    In particular, consider a column of the matrix $A$; it consists of the evaluations of a monomial
    $yx^i$ on all the points $(x_1, y_1), \cdots, (x_\ell,y_\ell)$.
    Using Eq.~(\ref{eq:ep}), we can write this column as the addition of a column $C_{i}$ that consists of the
    evaluation of the rational function $\frac{F(x)}{G(x)} x^i$ on the points $x_1, \cdots, x_\ell$ and a 
    column $\Gamma_{i}$ that consists of all the values $\gamma_{i,k}$ for $1 \le k \le \ell$.
    By the linearity of determinants,
    we can write the determinant of $A$ as the sum of determinants of two matrices  where 
    one matrix includes the column $C_{i}$ and the other has  $\Gamma_{i}$;
    observe that the magnitude of the determinant of the latter matrix can be upper bounded by
    $O(\varepsilon)$, with hidden constants that depend on $\Delta$.
    By performing this operation on all the columns, we can separate all the entries involving
    $\gamma_{i,k}$ into separate matrices and the magnitude of sum of the determinants can be
    bounded by $O(\varepsilon)$.
    
    Let $B$ be the matrix that remains after removing all the $\gamma_{i,k}$ terms. 
    We bound $|\det(B)|$.
    Note that $B$ consists of row vectors
    \[
        \begin{matrix}
            U = (1 & x & \cdots & x^{\D_1} & y & yx & \cdots & yx^{\D_G-1}).
        \end{matrix}
    \]
    evaluated at some value $x=x_k$ and $y=\frac{F(x_k)}{G(x_k)}$ at its $k$-th row. 
    This is equivalent to the evaluation of the following vector:
    \[
        \begin{matrix}
            (1 & x & \cdots & x^{\D_1} & \frac{F}{G} & \frac{F}{G}x & \cdots & \frac{F}{G}x^{\D_G-1}).
        \end{matrix}
    \]
    Observe that row $t$ of matrix $B$ will be evaluating $U$ on the point $x_k$. 
    Since $G(x_k) = \Theta(1) \not = 0$,  we can multiply row $k$ by $G(x_k)$ and this will only
    change the determinant by a constant factor. 
    With a slight abuse of the notation, let $B$ denote the matrix after this multiplication step. 
    Thus, the columns of $B$ now correspond  to the evaluation of the following vector. 
    \[
        \begin{matrix}
            (G & Gx & \cdots & Gx^{\D_1} & F & Fx & \cdots & Fx^{\D_G-1}).
        \end{matrix}
    \]
    Note that we can exchange columns and it will only flip the signs of the determinant of a matrix.
    We will focus on bounding the determinant of 
    \[
        \begin{matrix}
            (Gx^{\D_1} & Gx^{\D_1-2} & \cdots & G & Fx^{\D_G-1} & Fx^{\D_G-2}& \cdots & F).
        \end{matrix}
    \]    

	The key observation is that there is a strong connection between the Sylvester matrix of $G,F$ and matrix $B$.
	Recall that the Sylvester matrix of $G$ and $F$ is of the form
	\[
	\syl(G,F)=
	\begin{bmatrix}
		G_{\D_G}	&	G_{\D_G-1}	& \cdots			&	G_0			&	0				&	\cdots	&	0			&	0			\\
		0				&	G_{\D_G}	&	G_{\D_G-1}	& 	\cdots		&	G_0			&	\cdots	&	0			&	0			\\
		\vdots		&	\vdots		&	\vdots		&	\ddots		&	\vdots		&	\ddots	&	\vdots	&	\vdots	\\
		0				&	0				&	\cdots		&	G_{\D_G}	&	G_{\D_G-1}	&	\cdots	&	G_1		&	G_0		\\
		F_{\D_F}		&	F_{\D_F-1}	& \cdots			&	F_0			&	0				&	\cdots	&	0			&	0			\\
		0				&	F_{\D_F}		&	F_{\D_F-1}	& 	\cdots		&	F_0			&	\cdots	&	0			&	0			\\
		\vdots		&	\vdots		&	\vdots		&	\ddots		&	\vdots		&	\ddots	&	\vdots	&	\vdots	\\
		0				&	0				&	\cdots		&	F_{\D_F}		&	F_{\D_F-1}	&	\cdots	&	F_1		&	F_0		\\		
	\end{bmatrix},
	\]
	where $G_i$ (resp. $F_i$) is the coefficient of $x^i$ in $G$ (resp. $F$).
	Observe that 
	\begin{align*}
		\begin{matrix}
			(Gx^{\D_F-1} & Gx^{\D_F-2} & \cdots & G & Fx^{\D_G-1} & Fx^{\D_G-2}& \cdots & F)
		\end{matrix}
		=\\
		\syl(G,F) \cdot
		\begin{matrix}
			(x^{\D_F+\D_G-1}		&	x^{\D_F+\D_G-2}		&	\cdots	&	x		&	1)
		\end{matrix}^T,		
	\end{align*}
	which means that by the linear transformation described by $\syl(G,F)^{-1}$,
	which exists as $\res(G,F)=\det(\syl(G,F))\neq 0$,
	we can turn the last $\D_F+\D_G$ columns in $B$ to
	\[
		\begin{matrix}
			(x^{\D_F+\D_G-1}		&	x^{\D_F+\D_G-2}		&	\cdots	&	x		&	1)
		\end{matrix}.
	\]
	Since the remaining columns are all polynomials in $x$
	and the highest degree in column $i$ is $\D_1-i$
	for $i=0,1,\cdots,\D_F$,
	by using column operations, 
	we can eliminate all lower degree terms for each column
	and the only term left for column $i$ is $G_{\D_G}x^{\D_1-i}$.
	Note that column operations do not change the determinant.
	
	By assumption, the leading coefficients of $G$ is $1$, i.e., $G_{\D_G}=1$.
	Thus, this transforms $B$ into a Vandermonde matrix $V_B$ of size $\ell\times\ell$.
	By Theorem~\ref{thm:detvand},
	$|\det(V_B)|=\Omega(\lambda^{\ell^2})$.
	Since multiplying the inverse of $\syl(G,F)$
	scales $\det(B)$ by a factor of $\Theta(|\det(\syl(G,F)^{-1})|)=\Theta(|\res(G,F)^{-1}|)$,
	we bound $|\det(B)|=|\det(V_B)|/(1/|\res(G,F)|)=\Omega(|\res(G,F)|\lambda^{\ell^2})$.
  The claim then follows from this. 
\end{proof}


\section{Lower Bounds for Flat Intersection Reporting}

We are now ready to show lower bounds for flat intersection reporting.
We first establish a reduction from special polynomial slab reporting problems
to flat intersection reporting.

\subsection{A Reduction from Polynomial Slab Range Reporting to Flat-hyperslab Intersection Reporting}

	We study the following flat intersection reporting problem.

	\begin{definition}[Flat-hyperslab Intersection Reporting]
		In the $t$-flat-hyperslab intersection reporting problem,
		we are given a set $\sS$ of $n$ $(d-t)$-dimensional hyperslabs in $\R^d$,
        i.e., regions created by a linear translation of  $(d-t-1)$-flats,
		where $0 \le t < d$, as the input, and the goal is to preprocess $\sS$ into a data structure
		such that given any query $t$-flat $\gamma$,
		we can output $\sS\cap\gamma$,
		i.e., the set of $(d-t)$-hyperslabs intersecting the query $t$-flat, efficiently.
	\end{definition}
	
	First, observe that any $t$-flat that is not parallel to any of the axes can be formulated as 
	\[
		\begin{bmatrix}
			a_{0,1} 		& 			0 			& 		\cdots 		& 			0 			& 		0 				\\
			0 				& 			1 			&		\cdots 		& 			0 			& 		0 				\\		
			\vdots 		& 		\vdots 		&		\ddots 		& 		\vdots 		& 	\vdots 			\\
			0 				& 			0 			&		\cdots 		& 			1 			& 		0 				\\		
			a_{1,1}		&		a_{1,2}		&		\cdots		&		a_{1,t}		&	a_{1,t+1}		\\
			\vdots		&		\vdots		&		\ddots		&		\vdots		&	\vdots			\\
			a_{d-t,1}		&		a_{d-t,2}		&		\cdots		&		a_{d-1,t}		&	a_{d-t,t+1}		\\
		\end{bmatrix}
		\cdot
		\begin{bmatrix}
			\tau_1		\\
			\vdots	\\
			\tau_t			\\
			1
		\end{bmatrix}
		=
		\begin{bmatrix}
			x_1		\\
			\vdots	\\
			x_d
		\end{bmatrix}
		,
	\]	
	where $a_{i,j}$'s are the parameters defining the $t$-flat, 
	and $\tau_1,\cdots,\tau_t$ are the free variables that generate points in the $t$-flat.
	Note that we only need $(d-t)(t+1)$ independent $a_{i,j}$'s
	to define a $t$-flat.

	On the other hand, we consider $(d-t)$-hyperslabs of form
	\[
		\begin{bmatrix}
			1			&	0			&	\cdots	&	0			&	b_{1,1}		&	b_{1,2}		&	\cdots	&	b_{1,d-t}		\\
			0			&	1			&	\cdots	&	0			&	b_{2,1}		&	b_{2,2}		&	\cdots	&	b_{2,d-t}		\\
			\vdots	&	\vdots	&	\ddots	&	\vdots			&	\vdots		&	\vdots		&	\ddots	&	\vdots	\\
			0			&	0			&	\cdots	&	1			&	b_{t,1}		&	b_{t,2}		&	\cdots	&	b_{t,d-t}			\\
			0			&	0			&	\cdots	&	0			&	b_{t+1,1}	&	b_{t+1,2}	&	\cdots	&	b_{t+1,d-t}		\\
		\end{bmatrix}
		\cdot
		\begin{bmatrix}
			x_1		\\
			x_2		\\
			\vdots	\\
			x_{d-1}	\\
			x_d
		\end{bmatrix}
		=
		\begin{bmatrix}
			0		\\
			0		\\
			\vdots	\\
			0	\\
			-1+w
		\end{bmatrix},
	\]
	where $b_{i,j}$'s are the parameters defining a $(d-t-1)$-flat,
	and $w\in[0,w_0]$ is a parameter which adds one extra dimension to the flat to make it $(d-t)$-dimensional;
  	in essence, we will be considering all the $(d-t-1)$-flats for all $w \in [0,w_0]$ which will turn the query into
  	a $(d-t)$-hyperslab.
	
	Therefore, the intersection of a $t$-flat and a $(d-t)$-hyperslab must be a solution to
	\[
		\begin{bmatrix}
			1			&	0			&	\cdots	&	0			&	b_{1,1}	&	b_{1,2}	&	\cdots	&	b_{1,d-t}		\\
			0			&	1			&	\cdots	&	0			&	b_{2,1}	&	b_{2,2}	&	\cdots	&	b_{2,d-t}		\\
			\vdots	&	\vdots	&	\ddots	&	\vdots	&	\vdots	&	\vdots	&	\ddots	&	\vdots			\\
			0			&	0			&	\cdots	&	1			&	b_{t,1}	&	b_{t,2}	&	\cdots	&	b_{t,d-t}			\\
			0			&	0			&	\cdots	&	0	&	b_{t+1,1}	&	b_{t+1,2}	&	\cdots	&	b_{t+1,d-t}		\\
		\end{bmatrix}
		\cdot
		\begin{bmatrix}
			a_{0,1}	& 	0 				& 	\cdots 	& 	0 				& 	0 					\\
			0 					& 	1 				&	\cdots 	& 	0 				& 	0 					\\		
			\vdots 			& 	\vdots 		&	\ddots 	& 	\vdots 		& 	\vdots 			\\
			0 					& 	0 				&	\cdots 	& 	1 				& 	0 					\\		
			a_{1,1}			&	a_{1,2}		&	\cdots	&	a_{1,t}		&	a_{1,t+1}		\\
			\vdots			&	\vdots		&	\ddots	&	\vdots		&	\vdots			\\
			a_{d-t,1}			&	a_{d-t,2}		&	\cdots	&	a_{d-1,t}		&	a_{d-t,t+1}		\\
		\end{bmatrix}
		\cdot
		\begin{bmatrix}
			\tau_1		\\
			\tau_2		\\
			\tau_3 		\\
			\vdots	\\
			\tau_t			\\
			1
		\end{bmatrix}
		=
		\begin{bmatrix}
			0		\\
			0		\\
			\vdots	\\
			0	\\
			-1+w
		\end{bmatrix}.
	\]	
	Multiplying the two matrices, we obtain the following system
	\[
		\begin{bmatrix}
			a_{0,1}+\sum_{i=1}^{d-t}a_{i,1}b_{1,i}		&	\sum_{i=1}^{d-t}a_{i,2}	b_{1,i}		&	\cdots	&	\sum_{i=1}^{d-t}a_{i,t+1}	b_{1,i}	\\
			\sum_{i=1}^{d-t}a_{i,1}b_{2,i}					&	1+\sum_{i=1}^{d-t}a_{i,2}b_{2,i}		&	\cdots	&	\sum_{i=1}^{d-t}a_{i,t+1}b_{2,i}		\\
			\vdots												&	\vdots										&	\ddots	&	\vdots										\\
			\sum_{i=1}^{d-t}a_{i,1}b_{t+1,i}				&	\sum_{i=1}^{d-t}a_{i,2}b_{t+1,i}		&	\cdots	&	\sum_{i=1}^{d-t}a_{i,t+1}b_{t+1,i}	\\
		\end{bmatrix}
		\cdot
		\begin{bmatrix}
			\tau_1		\\
			\vdots	\\
			\tau_t			\\
			1
		\end{bmatrix}
		=
		\begin{bmatrix}
			0		\\
			\vdots	\\
			0	\\
			-1+w
		\end{bmatrix}.
	\]
	
	We denote this linear system by $A\vtau=\vs$ and assume 
	\begin{align}
		\det(A)\neq 0\label{eq:as1}
	\end{align}
	which is the case when the $t$-flat and the $(d-t)$-hyperslab properly intersect,
	and this system has a solution iff
	the last entry of the solution vector is $1$.
	So by Cramer's rule, we have
	\[
		1 = \frac{
					\begin{vmatrix}
						a_{0,1}+\sum_{i=1}^{d-t}a_{i,1}b_{1,i}		&	\sum_{i=1}^{d-t}a_{i,2}	b_{1,i}		&	\cdots	&	0			\\
						\sum_{i=1}^{d-t}a_{i,1}	b_{2,i}				&	1+\sum_{i=1}^{d-t}a_{i,2}b_{2,i}	&	\cdots	&	0			\\
						\vdots												&	\vdots										&	\ddots	&	\vdots	\\
						\sum_{i=1}^{d-t}a_{i,1}b_{t+1,i}				&	\sum_{i=1}^{d-t}a_{i,2}b_{t+1,i}		&	\cdots	&	-1+w		\\
					\end{vmatrix}
				}{
					\begin{vmatrix}
						a_{0,1}+\sum_{i=1}^{d-t}a_{i,1}b_{1,i}	&	\sum_{i=1}^{d-t}a_{i,2}b_{1,i}		&	\cdots	&	\sum_{i=1}^{d-t}a_{i,t+1}b_{1,i}		\\
						\sum_{i=1}^{d-t}a_{i,1}	b_{2,i}			&	1+\sum_{i=1}^{d-t}a_{i,2}b_{2,i}	&	\cdots	&	\sum_{i=1}^{d-t}a_{i,t+1}b_{2,i}		\\
						\vdots											&	\vdots									&	\ddots	&	\vdots										\\
						\sum_{i=1}^{d-t}a_{i,1}b_{t+1,i}			&	\sum_{i=1}^{d-t}a_{i,2}b_{t+1,i}	&	\cdots	&	\sum_{i=1}^{d-t}a_{i,t+1}b_{t+1,i}	\\
					\end{vmatrix}
				}.
	\]
	By the linearity of determinant, we have 
	\begin{equation}
		0 =
		\begin{vmatrix}
			a_{0,1}+\sum_{i=1}^{d-t}a_{i,1}b_{1,i}	&	\sum_{i=1}^{d-t}a_{i,2}b_{1,i}		&	\cdots	&	\sum_{i=1}^{d-t}a_{i,t+1}b_{1,i}				\\
			\sum_{i=1}^{d-t}a_{i,1}b_{2,i}				&	1+\sum_{i=1}^{d-t}a_{i,2}b_{2,i}	&	\cdots	&	\sum_{i=1}^{d-t}a_{i,t+1}b_{2,i}				\\
			\vdots											&	\vdots									&	\ddots	&	\vdots												\\
			\sum_{i=1}^{d-t}a_{i,1}b_{t+1,i}			&	\sum_{i=1}^{d-t}a_{i,2}b_{t+1,i}	&	\cdots	&	1+\sum_{i=1}^{d-t}a_{i,t+1}b_{t+1,i}-w		\\
		\end{vmatrix}\label{eq:reddet}.
	\end{equation}
	
  Consider the value of the above determinant using Leibniz formula for determinants, which is the sum of 
  $(t+1)!$ terms.
  Consider the terms that have at most 1 factor of $b_{i,j}$; these can only come from the diagonals. 
  Thus, any $t$-flat parameterized by $\ba=(a_{i,j})$
	intersects a query $(d-t)$-hyperslab parameterized by $\bb=(b_{i,j})$ if and only if
	\begin{align*}
		0		&= 	a_{0,1} + a_{0,1}\sum_{j=2}^{t+1}\sum_{i=1}^{d-1}a_{i,j} b_{j,i}
				+ \sum_{i=1}^{d-1}a_{i,1}b_{1,i}
				+	E(\ba,\bb) + f(\ba,\bb,w)
				&= P(\ba,\bb) + f(\ba,\bb,w),
	\end{align*}
	where $E(\ba,\bb)$ contains the sum of products of at least two distinct $a_{i_1,i_2}b_{i_3,i_1}$
	and $f(\ba,\bb,w)$ is a polynomial with factor $w$.
	
	Note that after fixing $\ba,\bb$, $f(\ba,\bb,w)$ is a polynomial in $w$ and we assume that
	\begin{align}
		\frac{\partial f(\ba,\bb,w)}{\partial w} 
		&=
		\begin{vmatrix}
			a_{0,1}+\sum_{i=1}^{d-t}a_{i,1}b_{1,i}	&	\sum_{i=1}^{d-t}a_{i,2}b_{1,i}		&	\cdots	&	\sum_{i=1}^{d-t}a_{i,t}b_{1,i}		\\
			\sum_{i=1}^{d-t}a_{i,1}b_{2,i}				&	1+\sum_{i=1}^{d-t}a_{i,2}b_{2,i}	&	\cdots	&	\sum_{i=1}^{d-t}a_{i,t}b_{2,i}		\\
			\vdots											&	\vdots									&	\ddots	&	\vdots										\\
			\sum_{i=1}^{d-t}a_{i,1}b_{t,i}				&	\sum_{i=1}^{d-t}a_{i,2}b_{t,i}		&	\cdots	&	1+\sum_{i=1}^{d-t}a_{i,t}b_{t,i}	\\
		\end{vmatrix} > 0. \label{eq:as2}
	\end{align}
	
	This implies the following lemma.
	\begin{lemma}
	\label{lem:redhelp}
		Assuming $\ba,\bb$ satisfying Assumptions~(\ref{eq:as1}) and (\ref{eq:as2}),
		for any fixed $\ba$, there is a $\bb$ such that $0\le P(\ba,\bb) \le -f(\ba,\bb,w_0)$
		if and only if there is some $w\in[0,w_0]$ such that $P(\ba,\bb)+f(\ba,\bb,w)=0$.
	\end{lemma}
	\begin{proof}
		This follows from $f(\ba,\bb,w)$ is a polynomial in $w$ and thus continuous
		and also $f(\ba,\bb,0)=0$ as $w$ is a factor of $f$ and decreasing in $[0,w_0]$ as $\frac{\partial f}{\partial w}>0$.
	\end{proof}
	
	Fixing $\ba$ in $P(\ba,\bb)$, we obtain a polynomial in $\bb$.
	Let $(P(\ba,\bb),f(\ba,\bb,w_0))=\{\bb:0\le P(\ba,\bb)\le-f(\ba,\bb,w_0)\}$ be a polynomial slab.
	This essentially establishes a reduction between polynomial slab reporting 
	and flat intersection reporting.
	
	\begin{corollary}
	\label{cor:reduction}
		Assuming $\ba,\bb$ satisfying Assumptions~(\ref{eq:as1}) and (\ref{eq:as2}),
		for any fixed $\ba$, there is a $\bb$ such that $\bb\in (P(\ba,\bb), f(\ba,\bb,w_0))$
		if and only if a $t$-flat parameterized by $\ba$ intersects a $(d-t)$-hyperslab parameterized by $\bb$.
	\end{corollary}

\subsection{Lower Bounds for Flat-hyperslab Intersection Reporting}

We are now ready to prove the lower bounds.
We show lower bounds for $1$-flat-hyperslab intersection reporting in $\R^d$
and $2$-flat-hyperslab intersection reporting in $\R^4$.

First observe by setting $t=1$ in Eq.~(\ref{eq:reddet}) and using Corollary~\ref{cor:reduction}
a polynomial slab reporting problem with polynomial
\begin{align}
    P_1(\ba,\bb)	&= 	a_{0,1} + a_{0,1}\sum_{i=1}^{d-1}a_{i,2} b_{2,i}
            + \sum_{i=1}^{d-1}a_{i,1}b_{1,i}
            +	\sum_{i,j=1\land i\neq j}^{d-1} (a_{i,1}a_{j,2}-a_{j,1}a_{i,2})b_{1,i}b_{2,j}\nonumber \\
                &=	b_{1,1}G_1(b_{2,2})+ F_1(b_{2,2}), \label{eq:p1}
\end{align}
reduces to a line-hyperslab intersection reporting problem, 
where to get $G_1$, we have collected all the monomials that have $b_{1,1}$ in them
and then we have factored $b_{1,1}$ out and we are considering it as a polynomial of
$b_{2,2}$ (all the other variables are considered ``constant''). 
$F_1$ is defined similarly by considering the remaining terms as a function of 
$b_{2,2}$. 
Observe that the polynomial does not have any term with degree 3. 
Let $G_1 = g_{1,1} b_{2,2} + g_{1,0}$ and 
$F_1 = f_{1,1} b_{2,2} + f_{1,0}$.

Similarly, polynomial slab reporting with
\begin{align}
    P_2(\ba,\bb) &= a_{0,1} + a_{0,1}\sum_{j=1}^2\sum_{i=2}^3a_{j,i}b_{i,j} + \sum_{j=1}^2a_{j,1}b_{1,j}\nonumber \\
    &+	a_{0,1}\sum_{j,l=1\land j\neq l}^2(a_{j,2}a_{l,3}-a_{j,3}a_{l,2})b_{2,j}b_{3,l}
    +	\sum_{j,l=1\land j\neq l}^2\sum_{k=2}^3(a_{j,1}a_{l,k}-a_{j,k}a_{l,1})b_{1,j}b_{k,l} \nonumber \\
    &= b_{1,1}G_2(b_{2,2})+ F_2(b_{2,2})\label{eq:p2}
\end{align}
reduces to $2$-flat-hyperslab intersection reporting in $\R^4$
where $G_2, F_2$ are defined similar as $G_1,F_1$.

For the moment, we focus on the case of line-hyperslab intersection reporting 
but the same applies also to $2$-flat-hyperslab intersection reporting in $\R^4$ 
since the polynomials $F_2$ and $G_2$
involved in the definition of Eq.~(\ref{eq:p2}) are quite similar
to Eq.~(\ref{eq:p1}).

Here, we would like to use our techniques from Section~\ref{sec:alg}. 
The general idea is that we will use Corollary~\ref{cor:fullslicing}, 
to reduce the $2(d-1)$-variate polynomials $P_1$ and $P_2$ into bivariate
polynomials on $b_{1,1}$ and $b_{2,2}$. 
Then, the variable $b_{1,1}$ will be our $y$ variable and $b_{2,2}$ will be the
$x$ variable in Section~\ref{sec:alg}, and $G_1$ and $F_1$ here will play the same role as in that section.
We will set 
\begin{align}
        a_{1,1} &= \frac{1+a_{1,2}a_{2,1}}{a_{2,2}} \label{eq:coef}
\end{align}
which will ensure that the leading coefficient of $G_1$ is 1.
This is our normalization step, since we can divide the equations defining
the intersection (and thus polynomials $P_1$ and $P_2$) by any constant. 
Eventually, the resultant of the polynomials $F_1$ and $G_1$ will play an important role. 
Observe that the resultant is 
\begin{align}
    \res(G_1,F_1)=\begin{vmatrix}
            1 & g_0 \\
             f_1 & f_0
         \end{vmatrix}= f_0 - g_0 f_1.    \label{eq:res}
\end{align}

\subsection{Construction of Input Points and Queries}
Now we are ready to describe our input and query construction. 
Assume we have a data structure that uses $S(n)$ space and has the query time
$Q(n) + O(k)$ where $k$ is the output size; for brevity we use $Q=Q(n)$.

We will start with a fixed line and a fixed hyperslab and then build the queries
and inputs very close to these two fixed objects.
However, we require a certain ``general position'' property with respect to these two fixed objects. 

Recall that Eq.~(\ref{eq:p1})  refers to the  condition of
whether a (query) line described by $\ba$ variables intersects
a $(d-2)$-dimensional flat described by the $\bb$ variables (which corresponds to setting the variable $w$ to zero).
Consider a fixed flat and a fixed line.
To avoid future confusion, let $\bA$ and $\bB$ refer to this fixed line and flat.
We require the following.
\begin{itemize}
    \item $\bA$ and $\bB$ must  intersect properly (i.e., the line is
        not contained in the flat).  Observe that it implies that when we consider
        $P_1(\bA,\bb)$ as a polynomial in $\bb$ variables, $\bB$ belongs to the
        zero set of $P_1(\bA,\bb)$.
        Note that this satisfies Assumption~(\ref{eq:as1}).
        
	\item The polynomial $P_1(\bA,\bb)$ (as a polynomial in $\bb$) is irreducible.
		This is true as long as $\bA$ is chosen so that no coefficient in $P_1$ is zero.
		To see this, note that $P_1$ is a polynomial in $\bb$ and any variable $b_{i,j}$ has degree $1$.
		Suppose for the sake of contradiction that $P_1$ is reducible, then the factorization must be of the form
		\[
			P_1(\bA,\bb) = \left( c_{10} + \sum_{i=1}^{d-1}c_{10}b_{1i} \right) \cdot \left( c_{20} + \sum_{i=1}^{d-1}c_{20}b_{2i} \right),
		\]
		for nonzero coefficients $c_{1i}, c_{2i}$.
		Then by Eq.~(\ref{eq:p1}),
		\begin{enumerate}
		\item $a_{0,1} = c_{10}c_{20}$,
		\item $\forall i = 1,\cdots, d-1: a_{0,1}a_{i,2} = c_{10}b_{2i}$,
		\item $\forall i = 1,\cdots, d-1: a_{i,1} = c_{1i}c_{20}$,
		\item $\forall i,j = 1,2,\cdots, d-1: a_{i,1}a_{j,2} - a_{j,1}a_{i,2} = c_{1i}c_{2i}$.
		\end{enumerate}
		However, for these conditions to hold, all coefficients of $P_1$ must be zero, a contradiction.
       
    \item Observe that the irreducibility of  $P_1(\bA,\bb)$ as a polynomial in $\bb$ implies that
        it has only finitely many points where the tangent hyperplane at those points is 
        parallel to some axis.  We assume $\bb$ is not one of those points. 

    \item The irreducibility of  $P_1(\bA,\bb)$ as a polynomial in $\bb$ can be used to satisfy Assumption~(\ref{eq:as2})
        since the corresponding polynomial of the determinant involved in Assumption~(\ref{eq:as2})
        can only have $\Theta(1)$ many common roots with $P_1(\bA,\bb)$.

    \item Finally, since the polynomial $P_1(\bA,\bb)$ is irreducible and since 
        $\res(G_1, F_1)$ is also  of degree 2 in $\bb$ variables, it follows that
        $\res(G_1, F_1)$ is algebraically independent of $P_1(\bA,\bb)$.
        This means that there are only finitely many places where both polynomials are zero,
        meaning, we can additionally assume that Eq.~(\ref{eq:res}) is non-zero
        (when evaluated at $\bB$).
\end{itemize}

Consider two parameters $\varepsilon_p$ and $\varepsilon_q = \varepsilon_p/C$ where
$C$ is a large enough constant and $\varepsilon_p$ is a parameter to be set later. 
Consider the parametric space of the input objects, where the variable
$\bb$ defines a single point. 
In such a space, $\bB$ defines a single point. 
Place an axis-aligned cube $\hR$ of side-length $\varepsilon_p$ centered around $\bB$.
The input slabs are defined by placing a set of $n$ random points inside 
$\hR$.
Each point in $\hR$ defines a $(d-2)$-dimensional flat.
We set $w=\Theta(\frac{Q}{n})$ which in turn defines a ``narrow $(d-1)$-hyperslab''.  

We now define the set of queries. 
Notice that $P_1$ has exactly $2(d-1)$
algebraically independent coefficients;
these are the coefficients of linear terms involved plus $a_{0,1}$;
recall that by Eq.~(\ref{eq:coef}), $a_{1,1}$ was fixed as a function of $a_{1,2}a_{2,1}$ and
$a_{2,2}$ but we still have $a_{0,1}$ as a free parameter.
These $2(d-1)$ coefficients define another parametric space, where $\bA$ 
denotes a single point. 
Place a $2(d-1)$-dimensional hypercube of side length $\varepsilon_q$ and then subdivide it into
a grid where the side-length of every cell is $\tau$. 
Every grid point now defines a different query.
Let $\Q$ be the set of all the queries we have constructed.

Notice that a query defined by a point $\ba \in \Q$ defines a line in the primal space,
but when considered in the parametric space $\hR$, it corresponds to a manifold (zeroes of a degree two
multilinear polynomial)
that includes the set of points that correspond to 
$(d-2)$-dimensional flats that pass through the line in the primal space.
The variable $w$ allows us to turn it to a range reporting problem where
we need to output any $(d-2)$-dimensional flat that passes within $w$ vertical distance
of the query line. 
The following observations and lemmas are the important geometric properties
that we require out of our construction. 

\begin{observation}\label{ob:initdiff}
    For two different queries $\ba_1$, and $\ba_2$, the polynomials $P_1(\ba_1,\bb)$
    and $P_1(\ba_2,\bb)$ differ by at least $\tau$ in at least one of their coefficients.
\end{observation}

\begin{observation}\label{ob:line}
    Consider a line $f$ parallel to an axis. 
    For small enough $\varepsilon_p$,  and any $\ba \in \Q$,
    the function $P_1(\ba,\bb)$ evaluated on the line $f$
    is such that the magnitude of its derivative is
    bounded by $\Omega(1)$.
\end{observation}
\begin{proof}
    Recall that $\bB$ was chosen such that the manifold corresponding to $\bA$ 
    does not have a tangent parallel to any of the axes at point $\bB$ and thus
    the derivate of the function $P_1(\bA,\bB)$ is non-zero at $\bB$.
    The lemma then follows since $\varepsilon_p$ and $\varepsilon_q$ are small enough and
    $P_1(\bA,\bB)$ is a continuous function w.r.t any of its variables.
\end{proof}

\begin{observation}\label{ob:vol}
    The intersection volume of the range defined by a query $\ba$
    and $\hR$ is $\Theta(w\vol(\hR))$ if $C$ in the definition of 
    $\varepsilon_q$ is large enough, for $w \le \varepsilon_p$.
\end{observation}
\begin{proof}
    Observe that the query manifold defined by $\bA$ passes through the center, $\bB$, of
    $\hR$ by construction.
    Since each coordinate of $\ba$ differs from $\bA$ by at most $\varepsilon_q$,
    it thus follows that by setting $C$ large enough, we can ensure that the distance
    between $\bB$ and $\ba$ is less than $\varepsilon_p/2$. 
    The claim now follows by integrating the volume over vertical lines
    using Observation~\ref{ob:line}.
\end{proof}

\begin{lemma}\label{lem:proj}
    Consider a query $\ba \in \Q$ and let $\rr$ be the range
    that represents $\ba$ in the parametric space defined by $\hR$.
    Consider an interval $\I$ on the $i$-th side of $\hR$, for some $i$.
    Let $\rr_\I$ be the subset of $\rr$ whose projection on the $i$-th side of
    $\hR$ falls inside $\I$. 
    Then, the volume of $\rr_\I$ is $O(\vol(\hR) w |\I|/\varepsilon_p)$.
    
    In addition, for any 2D flat $f$, the area of the intersection of 
    $f$ and $\rr$ is $O(w \varepsilon_p)$.
\end{lemma}
\begin{proof}
    Both claims follow through Observation~\ref{ob:line} by integrating 
    the corresponding volumes over lines parallel to axes.
\end{proof}
\subsection{Using the Framework}

Observe that by the above Observation~\ref{ob:vol}, setting $w=\Theta(\frac{Q}{n})$  satisfies
Condition~\ref{cond:cond1} of the lower bound framework in Theorem~\ref{thm:rrfw}.

Satisfying Condition~\ref{cond:cond2} requires a bit more work however.
To do that, consider two queries defined by points $\ba_1$ and $\ba_2$.
Let $\rr_1$ and $\rr_2$ be the two corresponding ranges
in the parametric space of $\hR$.

To satisfy Condition~\ref{cond:cond2}, 
assume for contradiction that the volume of $\rr_1 \cap \rr_2$ 
is large, i.e.,
larger than $\vol(\hR)/(Q\psi)$ where $\psi=2^{\sqrt{\log n}}$.
We now combine Observation~\ref{ob:initdiff}, and Corollary~\ref{cor:fullslicing} with
parameter $\vartheta$ set to $\varepsilon_0\frac{\varepsilon_p}{Q(n) \psi}$ where
$\varepsilon_0$ is a small enough constant and where 
$X_1$ represents $b_{1,1}$, $X_2$ represents $b_{2,2}$ and the remaining indeterminates
represent the rest of variables in $\bb$;
note that the value of $d$ in Corollary~\ref{cor:fullslicing}
is $\beta=2(d-1)$ and $\U=\binom{2+1}{2}=3$.
Observe that each interval $\I_i$ determined by Corollary~\ref{cor:fullslicing} defines
a slab parallel to the $i$-th axis in $\hR$; let $\hRb$ be the union of these slabs.
By Lemma~\ref{lem:proj}, and choice of small enough $\varepsilon_0$, 
a positive fraction of the intersection volume of
$\rr_1$ and $\rr_2$ must lie outside $\hRb$.
In addition, Corollary~\ref{cor:fullslicing} allows us to pick some fixed values 
for all variables in $\bb$, except for $b_{1,1}$ and $b_{2,2}$ with the property 
the final polynomials 
$H_1$ and $H_2$ (on indeterminates
$b_{1,1}$ and $b_{2,2}$) that we obtain have  the property that
they have at least one coefficient which 
differs by at least 
\begin{align}
    \Omega\left( \tau \left(\varepsilon_0\frac{\varepsilon_p}{Q(n) \psi}  \right)^{3(\beta-2)} \right)\label{eq:finaldiff}
\end{align}
between them; we call this operation of plugging values for all $\bb$ except for
$b_{1,1}$ and $b_{2,2}$ \textit{slicing}.
After slicing, we are reduced to the bivariate case; consider
the set of points on which both polynomials $H_1$ and $H_2$ have value $O(w)$.
If the 2D area of such points is $o(\varepsilon_p / (Q\psi)$, 
    then we call this a \textit{good} slice, otherwise, it is a \textit{bad} slice.
    By Lemma~\ref{lem:proj}, we can observe that there must be bad slices since
    if all the slices are good, by integration of the area of
    the intersection of $\rr_1$ and $\rr_2$ over all the remaining
    variables in $\bb$, will yield that the volume of the intersection of
    $\rr_1$ and $\rr_2$ smaller is $o(\vol(\hR)/(Q\psi))$ which is a contradiction.

We now show that we can arrive at a contradiction, assuming the existence of a bad slice.
Given a bad slice, and any constant $\ell$, we can find
$\ell$ points $(x_1, y_1), \ldots, (x_\ell,y_\ell)$ such that 
$|x_{k_1}-x_{k_2}| \ge \Omega(\varepsilon_p/(Q\psi))$ for all $1\le k_1< k_2\le \ell$ and
that $H_1(x_k,y_k), H_2(x_k,y_k) = O(w)$ for all $k\in\{1,2\cdots,\ell\}$.
Observe that $H_i(x,y)$ has only monomials $y$, $x$, $xy$ and a constant term.
The critical observation here is that the coefficient of the monomial $xy$ is always
1 since the coefficient of the monomial $b_{1,1}b_{2,2}$ was 1 and there was no monomial of
degree three in $P_1$, meaning, after slicing this coefficient will not change. 
We pick $\ell=3$ and thus we tweak all the three other coefficients of
$H_1$.
Tweaking $H_1$ such that $\tilde{H}_1(x_k,y_k) = H_2(x_k,y_k)$ corresponds to solving a 
linear system of equations that come from evalutions of monomials $X$, $Y$, and a constant term
at points $(x_k, y_k)$. 
We can thus use
Lemma~\ref{lem:detX} with $\Delta_1 = \Delta_G =1$, $\lambda = \Omega(\varepsilon_p/(Q\psi))$.
Observe that $\res(G,F)$ here is a constant by the properties of our construction. 
Also observe that by Lemma~\ref{lem:detX}, the magnitude of the determinant of matrix $A$ defined in
Lemma~\ref{lem:detX} is 
\[
\Omega \left( \left(  \varepsilon_p/(Q\psi)) \right)^9 \right).
\]
By the same argument in~\cite{ac22}, this means that the tweaking operation can be done
such that each coefficient of $H_1$ is changed by at most
\begin{align}
O\left( \left(  \varepsilon_p/(Q\psi)) \right)^{-9}w \right)\label{eq:finalt}.
\end{align}
We observe that after tweaking,
$\tilde{H}_1$ and $H_2$ must coincide since  by Lemma~\ref{lem:detX}, the determinant of
the relevant monomials is non-zero and thus there's a unique polynomial that passes through 
points $(x_1, y_1), \cdots, (x_\ell, y_\ell)$.
Finally, to get a contradiction, we simply need to ensure that 
Eq.~(\ref{eq:finalt}) is asymptotically smaller than Eq.~(\ref{eq:finaldiff}).
This yields a bound for the value of $\tau$,
\begin{align}
    \tau = \Theta\left(w (Q\psi)^{3(\beta-2)+9} \right) =  \Theta\left(w (Q\psi)^{3\beta+3} \right)
\end{align}
where we have assumed that $\varepsilon_p$, and $\varepsilon_0$ are small enough constants
that have been absorbed in the $\Omega(\cdot)$ notation.
Thus, this choice of $\tau$ will make sure that Condition~\ref{cond:cond2} of the framework is also 
satisfied. 
It remains to calculate the number of queries that have been generated.
Observe that $\tau$ was the side-length of a small enough grid around the point $\bA$
in a $\beta$-dimensional space.
Thus, the number of queries we generated is
\begin{align}
    m=\domega\left( \left(  \frac{1}{\tau} \right)^{\beta} \right) = \domega\left( \frac{n^{\beta}}{Q^{\beta(3\beta+4)}}\right).
\end{align}
Applying Theorem~\ref{thm:rrfw} yields a space lower bound of 
\begin{align}
  S(n) =  \domega(mQ) =   \domega\left( \frac{n^{2(d-1)}}{Q^{4(3d-1)(d-1)-1}}\right)
\end{align}
for line-hyperslab intersection reporting since $\beta=2(d-1)$.
One can verify that the same argument works for triangle-triangle intersection reporting in $\R^4$, since
$P_2$ is also a multilinear polynomial of degree two.
In this case, $\beta=6$ which yields 
a space lower bound of
\begin{align}
  S(n) =     \domega\left( \frac{n^{6}}{Q^{125}}\right).
\end{align}

To sum up, we obtain the following results:

	\begin{theorem}
	\label{thm:line-hyperslab}
		Any data structure that solves line-hyperslab intersection reporting in $\R^d$
		must satisfy a space-time tradeoff of 
		$S(n)=\domega\left(\frac{n^{2(d-1)}}{Q(n)^{(4(3d-1)(d-1)-1}}\right)$.
	\end{theorem}
	
	\begin{theorem}
	\label{thm:tri-tri}
		Any data structure that solves triangle-triangle intersection reporting in $\R^4$
		must satisfy a space-time tradeoff of 
		$S(n)=\domega\left(\frac{n^{6}}{Q(n)^{125}}\right)$.
	\end{theorem}


\section{Conclusion and Open Problems}

We study line-hyperslab intersecting reporting in $\R^d$
and triangle-triangle intersecting reporting in $\R^4$.
We show that any data structure with $n^{o(1)}+O(k)$ query time
must use space $\domega(n^{2(d-1)})$ and $\domega(n^6)$
for the two problems respectively.
This matches the classical upper bounds for the small $n^{o(1)}$ query time case
for the two problems and answer an open problem for lower bounds
asked by Ezra and Sharir~\cite{es22i}.
Along the way, we generalize and develop the lower bound technique used in~\cite{ac21,ac22}.

The major open problem is how to show a lower bound for
general intersection reporting between objects of $t$ and $(d-t)$ dimensions
or for flat semialgebraic objects as studied recently in~\cite{aaeks22}.
Many of our techniques work, however, one big challenge is that 
after applying Corollary~\ref{cor:fullslicing},
the leading coefficient changes and thus 
we can no longer guarantee big gaps between coefficients.

\bibliography{reference}{}
\bibliographystyle{abbrv}

\begin{appendices}

\section{Proof of Lemma~\ref{lem:gen1d}}
\label{sec:proof-gen1d}
\genbase*

	\begin{proof}
	The proof is by contradiction.
	Assume for the sake of contradiction that $|\I|=\omega((w/\eta)^{1/\U})$.
	We pick $\D+1$ points $(u_1,v_1),\cdots,(u_{\D+1},v_{\D+1})$ from $y=P(x)$ 
	for $u_1,\cdots,u_{\D+1}\in\I$ such that $|u_{k_1}-u_{k_2}|=\omega((w/\eta)^{1/\U})$ for $k_1\neq k_2$.
	Let $(u_1,v_1+\xi_1),\cdots,(u_{\D+1},v_{\D+1}+\xi_{\D+1})$ be $\D+1$ points on $y=Q(x)$.
	By definition $|\xi_k|\le w$ for all $k=1,2,\cdots,\D+1$.
	 We would like to tweak coefficient $a_i$ of $P(x)$ by $\delta_i$ for $i=0,1,\cdots,\D$
	 to obtain a polynomial $P'(x)$ such that 
	 $y=P'(x)$ and $y=Q(x)$ agree on $x=u_1,\cdots,u_{\D+1}$.
     Observe that to do that, for each $k$, we would like to have 
	 \[
	 	v_k+\xi_k=\sum_{i=0}^{\D}(a_i+\delta_i)u_k^i 
		\implies	\xi_k=\sum_{i=0}^{\D}\delta_i u_k^i
	\]
    where the last follows from $v_k = \sum_{i=0}^{\D}a_iu_k^i$.
    Thus, to perform the tweaking, 
	each $\delta_i$ should satisfy the following system
	\[
		\begin{bmatrix}
		1 			&		x_1 			&		\cdots 		&		x_1^\D 			\\
		\vdots 	&		\vdots		&		\ddots		&		\vdots			\\
		1 			&		x_{\D+1}		&		\cdots 		&		x_{\D+1}^\D 	\\
		\end{bmatrix}
		\cdot
		\begin{bmatrix}
		\delta_0		\\
		\vdots		\\
		\delta_{\D}	\\
		\end{bmatrix}
		=
		\begin{bmatrix}
		\xi_1			\\
		\vdots		\\
		\xi_{\D+1}	\\
		\end{bmatrix},
	\]
	or $A\cdot\vdelta = \vxi$.
	Note that $A$ is a Vandermonde matrix.
	So by Theorem~\ref{thm:detvand},
    $\det(A)=\Omega(|\I|^{\U})$, since $|x_{k_1} - x_{k_2} | = \Omega(|\I|)$. 
    This shows that there exists a unique solution for $\delta_i$ and thus
    we can perform the tweaking.
    In addition, it follows that $|\delta_i|  = O( \frac{\xi}{\det(A)})= O( \frac{w}{\det(A)})$. 
    Now observe that if we assume $\I = \omega((w/\eta)^{1/\U})$, it follows
    that $|\delta_i|=o(\eta)$.
	Since $P',Q$ have $\D+1$ points in common,
	they must be equivalent. 
	This mean $|a_i-b_i|=o(\eta)$ for all $0\le i\le \D$,
	a contradiction.
	\end{proof}

\section{Proof of Lemma~\ref{lem:slicing}}
\label{sec:proof-slicing}
\slicing*

	\begin{proof}
		Note that
		\[
			P_1=\sum_{\vi\in I} A_\vi X^\vi=\sum_{\vj\in I_{:d-1}}h_\vj(X_d)X_{:d-1}^\vj,
		\]
		and
		\[
			P_2=\sum_{i\in I} B_\vi X^\vi=\sum_{\vj\in I_{:d-1}}\hat{h}_\vj(X_d)X_{:d-1}^\vj,
		\]
		where $I_{:d-1}$ stands for the set of $(d-1)$-tuples formed by taking the $(d-1)$ 
		entries of every $d$-tuple in $I$ and $X_{:d-1}=(X_1,\cdots,X_{d-1})$,
		and the coefficients of $X_{:d-1}$ are polynomials in $X_d$ of form
		\[
			h_\vj(X_d) = \sum_{i=0}^{\D} A_{\vj\oplus i}X_d^i,
		\]
		and
		\[
			\hat{h}_\vj(X_d) = \sum_{i=0}^{\D} B_{\vj\oplus i}X_d^i,
		\]
		where $\vj\oplus i$ stands for a $d$-tuple formed by appending $i$ to $\vj$.
		Since $P_1\not\equiv P_2$, there must exists one $\vj$ such that $h_{\vj}(X_d)-\hat{h}_{\vj}(X_d)\not\equiv 0$.
		Then by Lemma~\ref{lem:gen1d}, the interval length for $X_d$ in which 
		$|h_{\vj}(X_d)-\hat{h}_{\vj}(X_d)|\le\eta_{d-1}$
		is upper bounded by $O(\eta_{d-1}/\eta_d)^{1/\U})$.	
		By a union bound over all monomials, the lemma follows.	
	\end{proof}

\end{appendices}

\end{document}